\documentclass{article}

\usepackage[latin1]{inputenc}
\usepackage[english]{babel} 
\usepackage{amsmath,amsfonts,amssymb,amsthm}

\newtheorem{theorem}{Theorem}

\newtheorem{lemma}[theorem]{Lemma}
\newtheorem{proposition}[theorem]{Proposition}

\theoremstyle{definition}
\newtheorem{definition}[theorem]{Definition}

\newcommand{\norm}[1]{\left\Vert {#1} \right\Vert}
\newcommand{\abs}[1]{\left\vert {#1} \right\vert}
\newcommand{\set}[1]{\left\{ {#1} \right\}}
\newcommand{\prt}[1]{\left( {#1} \right)}

\newcommand{\scal}[1]{\left< {#1} \right>}

\newcommand{\setN}{{\mathbb N}}              
\newcommand{\setZ}{{\mathbb Z}}
\newcommand{\setR}{{\mathbb R}}
\newcommand{\setC}{{\mathbb C}}

\newcommand{\lequi}{\ \ \Longleftrightarrow\ \ }

\newcommand{\A}{\mathcal{A}}
\newcommand{\B}{\mathcal{B}}

\renewcommand{\H}{\mathcal{H}}
\newcommand{\I}{\mathcal{I}}
\newcommand{\J}{\mathcal{J}}
\renewcommand{\L}{\mathcal{L}}
\newcommand{\M}{\mathcal{M}}

\newcommand{\T}{\mathcal{T}}

\newcommand{\Sw}{\mathcal{S}}

\begin{document}

\title{\vspace{-1.1cm}{\bf Temporal Lorentzian Spectral Triples}}

\author{Nicolas Franco}

\date{{\footnotesize 
Copernicus Center for Interdisciplinary Studies{\footnote{{\small supported by a grant from the John Templeton Foundation}}},\\
Jagiellonian University,\\
ul. S{\l}awkowska 17, PL-31-016 Krak\'ow, Poland\\[0.3cm]
University of Namur, Department of Mathematics,\\
Rempart de la Vierge 8, B-5000 Namur, Belgium\\[0.3cm]
nicolas.franco@math.unamur.be
}}

\maketitle

\begin{abstract}
We present the notion of temporal Lorentzian spectral triple which is an extension of the notion of pseudo-Riemannian spectral triple with a way to ensure that the signature of the metric is Lorentzian. A temporal Lorentzian spectral triple corresponds to a specific 3+1 decomposition of a possibly noncommutative Lorentzian space. This structure introduces a notion of global time in noncommutative geometry. As an example, we construct a temporal Lorentzian spectral triple over a Moyal--Minkowski spacetime. We show that, when time is commutative, the algebra can be extended to unbounded elements. Using such an extension, it is possible to define a Lorentzian distance formula between pure states with a well-defined noncommutative formulation.
\end{abstract}

\section{Introduction}

The theory of noncommutative geometry is related to the duality existing between algebra and geometry. Thanks to Gel'fand--Naimark theorem, we know that unital commutative $C^*$-algebras are equivalent in a category-theoretical sense to algebras of continuous functions on compact Hausdorff spaces. A similar result is valid for locally compact Hausdorff spaces being equivalent to not necessarily unital commutative algebras. The extension of this duality to noncommutative $C^*$-algebras opens the way to the definition of  noncommutative spaces.

Behind those topological considerations, some metric information can be added by using Connes' notion of spectral triples \cite{C94,MC08,Var}.  A spectral triple $(\A,\H,D)$ is composed of a Hilbert space $\mathcal{H}$, a unital pre-$C^*$-algebra $\mathcal{A}$ with a faithful representation as bounded multiplicative operators on $\mathcal{H}$ and a self-adjoint operator $D$ densely defined on $\mathcal{H}$ with compact resolvent and such that all commutators $[D,a]$ are bounded for every $a\in\A$.

In the commutative case, the correspondence with Riemannian geometry is given by the so-called reconstruction theorem \cite{C96,C08}. For every compact Riemannian spin manifold, a spectral triple can be constructed by taking the Hilbert space $\H = L^2(\M,S)$ of square integrable spinor sections over $\M$, the unital pre-$C^*$-algebra $\A=C^\infty(\M)$ (with pointwise multiplication and supremum norm) and the Dirac operator $D = -i(c \circ \nabla^S)$, where $c$ represents the Clifford action. The reconstruction theorem states that, under some additional axioms, commutative unital spectral triples correspond to existing compact Riemannian spin manifolds. Information on the differential structure can be extracted using the data given by a spectral triple. As an example, a notion of Riemannian distance between two states $d(\xi,\eta)$, corresponding to the usual notion of distance, can be defined as $d(\xi,\eta)=\sup\set{  \abs{\xi(a)-\eta(a)} : a \in \A,\  \norm{[D,a]} \leq 1 }$.

One goal of Connes' noncommutative geometry is to provide a new mathematical background which would be suitable to describe the fundamental forces of physics in a similar formalism. Such a background is given by the product of two spectral triples, a commutative one representing Euclidean gravitation (Euclidean in the sense of positive signature), and a noncommutative one representing a (still classical) standard model \cite{MC08,MC207}.

At this time, the theory of spectral triples is only completely defined for spaces with positive signature for the metric. However, this does not correspond to physical reality. A Lorentzian counterpart would be very useful especially for models dealing with gravitation. The extension to pseudo-Riemannian and Lorentzian spaces of the notion of a spectral triple is far from being complete, but some propositions exist \cite{Haw,Kopf98,Kopf00,Kopf01,Kopf02,Stro,Suij,PasV,Pas,Rennie12}.

In this paper we present a new way to construct Lorentzian spectral triples, starting from the interesting proposition by Strohmaier \cite{Stro} on building pseudo-Riemannian spectral triples based on Krein spaces. We show that it is possible to construct a spectral triple representing a (possibly noncommutative) globally hyperbolic Lorentzian manifold with a specific "3+1" decomposition and with a specific element of the algebra representing a global time. Such a construction is called a temporal Lorentzian spectral triple. The time element has the particularity to completely define a fundamental symmetry needed to recover the self-adjointness of the Dirac operator. As an improvement of the pseudo-Riemannian approach \cite{Stro,Rennie12}, we show that a specific construction of a fundamental symmetry, defined as the commutator between the Dirac operator and a time element, ensures that the signature of the metric corresponds to a Lorentzian space. This is a crucial property in order to recover the information coming from the Lorentzian nature of physical spacetimes like causality or Lorentzian distance as in \cite{F6}.

Furthermore, the definition of a Lorentzian distance formula as constructed in \cite{F6,F3} requires the use of unbounded functions and its noncommutative generalization cannot fit into the usual formalism of $C^*$-algebras. Starting with our axioms of temporal Lorentzian spectral triple, we show that, when time is commutative, such a construction admits an extension of the algebra to unbounded elements by the way of a filtration. Then we show that a well-defined noncommutative formulation of a Lorentzian distance formula is possible in this context by defining a notion of distance between the unique extension of two given pure states.

The plan of this paper is the following. In Section \ref{Krein} we recall the basic notions about Krein spaces and pseudo-Riemannian spectral triples. In Section \ref{TLSP} we present the axioms of a temporal Lorentzian spectral triple and the method for generating a filtered algebra of unbounded elements. In Section \ref{NLMP} we construct an example of noncommutative temporal Lorentzian spectral triple as the Moyal--Minkowski spacetime. In Section \ref{distsec}, we illustrate how the extension to unbounded elements can be used to define a Lorentzian distance formula for noncommutative temporal Lorentzian spectral triples with commutative time.

\section{Krein Spaces and Pseudo-Riemannian Spectral Triples}\label{Krein}

We start with a review of the tools introduced by Strohmaier in order to adapt the construction of spectral triples to pseudo-Riemannian manifolds. We will just be interested here in the results and we refer the reader to \cite{Stro} for the proofs and details. Some results about Dirac operators in pseudo-Riemannian geometry can also be found in \cite{BaumG,BaumE}.

In the same way as a Riemannian spectral triple, a pseudo-Riemannian spectral triple is a triple \mbox{$\prt{\A,\H,D}$} which corresponds in the commutative case to the algebra \mbox{$\A = C^\infty(\M)$} of smooth functions over a pseudo-Riemannian spin manifold $\M$ of signature \mbox{$(p,q)$} (with $q\geq 1$), to the Hilbert space $\H$ consisting of square integrable sections of the spinor bundle over $\M$ and on which there exists a representation of $\A$ as multiplicative bounded operators, and to the Dirac operator \mbox{$D = -i(c \circ \nabla^S)$} acting on the space $\H$.

The Hilbert space $\H$ is endowed with the positive definite Hermitian structure
\[(\psi,\phi) = \int_\M \psi^*\phi\, d\mu_g\]
where $d\mu_g = \sqrt{\abs{\det g}} \;d^nx$ is the pseudo-Riemannian density on $\M$.

However, this structure does not admit any self-adjoint pseudo-Riemannian Dirac operator related to this positive definite inner product. Instead, the Dirac operator corresponding to the pseudo-Riemannian metric is an essentially Krein-self-adjoint operator if we transform $\H$ into a Krein space \cite{Bog}, which is a space with indefinite inner product.

\begin{definition}
An indefinite inner product on a vector space $V$ is a map \mbox{$V \times V \rightarrow \mathbb C$} which satisfies:
\[(v,\lambda w_1 + \mu w_2) = \lambda (v,w_1) + \mu(v, w_2), \qquad \overline{(v,w)} = (w,v).\]
An indefinite inner product is nondegenerate if
\[(v,w)=0 \quad\forall v\in V \ \ \Rightarrow\ \  w = 0.\]
\end{definition}

Let us suppose that $V$ can be written as the direct sum of two orthogonal spaces \mbox{$V = V^+ \oplus V^-$} such that the inner product is positive definite on $V^+$ and negative definite on $V^-$. Then the two spaces $V^+$ and $V^-$ are two pre-Hilbert spaces by the induced inner product (with a multiplication by $-1$ on the inner product for the second one).

\begin{definition}[\cite{Bog}]
If the two subspaces $V^+$ and $V^-$ are complete in the norm induced on them and if the indefinite inner product on $V$ is nondegenerate, then the space $V = V^+ \oplus V^-$ is called a Krein space. The  indefinite inner product is called a Krein inner product.
\end{definition}

\begin{definition}[\cite{Bog}]
For every decomposition \mbox{$V = V^+ \oplus V^-$} the operator $\J = \text{id}_{V^+} \oplus -\text{id}_{V^-}$ respecting the property $\J^2 = 1$ is called a fundamental symmetry. Such an operator defines a positive definite inner product on $V$ by \mbox{$\scal{\,\cdot\,,\,\cdot\,}_\J = \prt{\,\cdot\,,\J\,\cdot\,}$}.
\end{definition}

Each fundamental symmetry of a Krein space $V$ defines a Hilbert space structure, and two norms associated with two different fundamental symmetries are equivalent. So it is natural to define the space of bounded operators as the space of bounded operators on the Hilbert space defined for any fundamental symmetry.

\begin{definition}[\cite{Bog}]
If $A$ is a densely defined linear operator on $V$, the Krein-adjoint $A^+$ of $A$ is the adjoint operator defined for the Krein inner product \mbox{$(\,\cdot\,,\,\cdot\,)$} (with the usual definition and domain). An operator $A$ is called Krein-self-adjoint if \mbox{$A=A^+$}.
\end{definition}

Of course, for any fundamental symmetry $\J$ we can define an adjoint $A^*$ for the positive definite inner product $\scal{\,\cdot\,,\,\cdot\,}_\J$. In this case, the Krein-adjoint is related to it by \mbox{$A^+ = \J A^* \J$}, and an operator $A$ is Krein-self-adjoint if and only if $\J\!A$ or $A\,\J$ are self-adjoint in the Hilbert space.

Now the question is how we could define a Krein space structure from a spin manifold with pseudo-Riemannian metric. This is done by using spacelike reflections.

\begin{definition}[\cite{Stro}]
A spacelike reflection $r$ is an automorphism of the vector bundle $T\M$ such that:
\begin{itemize}
\item $g(r\,\cdot\,,r\,\cdot\,) = g(\,\cdot\,,\,\cdot\,)$,
\item $r^2=\text{id}$,
\item $g^r(\,\cdot\,,\,\cdot\,) = (\,\cdot\,,r\,\cdot\,)$ is a positive definite metric on $T\M$.
\end{itemize}
\end{definition}

It is clear that, for every pseudo-Riemannian metric of signature $(p,q)$, the tangent bundle can be split into an orthogonal direct sum \mbox{$T\M = T\M_+^p \oplus T\M_-^q$} where the metric is positive definite on the $p$-dimensional bundle $T\M_+^p$ and negative on the $q$-dimensional bundle $T\M_-^q$, thus a spacelike reflection is automatically associated by defining \mbox{$r(v_{+|x} \oplus v_{-|x}) = v_{+|x} \oplus -v_{-|x}$}. This splitting is  transposed to the cotangent bundle \mbox{$T\M^* = T\M_+^{*p} \oplus T\M_-^{*q}$} by isomorphism.

\begin{proposition}[\cite{Stro}]\label{propslr}
For each spacelike reflection $r$, there is an associated fundamental symmetry $\J_r$ defined from the Clifford action $c$ on a local oriented orthonormal basis \mbox{$\set{e_1,e_2,\dots,e_q}$} of $T\M_-^{*q}$ by:
\[\J_r = i^{\frac{q(q+1)}{2}} c(e_1)c(e_2)\dots c(e_q) =  i^{\frac{q(q+1)}{2}}\gamma^1\gamma^2\dots\gamma^q.\]
\end{proposition}

This definition is independent of the choice of the local basis. With such a fundamental symmetry, the space  $\H$  of square integrable sections of the spinor bundle becomes a Krein space endowed with the indefinite inner product
\[(\psi,\phi) = \int_\M \psi^* \J \phi\, d\mu_g.\]
 Actually, this operation is similar to a Wick rotation, but performed at an algebraic level.
 
In the special case of a $4$-dimensional Lorentzian manifold, with signature \mbox{$(-,+,+,+)$}, a canonical fundamental symmetry is just given by \mbox{$J = i\gamma^0$}.\footnote{With our choice of signature, the flat Dirac matrix $\gamma^0$ is such that $\prt{\gamma^0}^2=-1$ and $\prt{\gamma^0}^*=-\gamma^0$, so $J = i\gamma^0$ respects the conditions of a fundamental symmetry. The other flat Dirac matrices respect $\prt{\gamma^a}^2=1$ and $\prt{\gamma^a}^*=\gamma^a$ for $a=1,2,3$.}

We have then the following result concerning the Dirac operator:

\begin{proposition}[\cite{Stro,BaumG}]\label{cpt1}
If there exists a spacelike reflection such that the Riemannian metric $g^r$ associated is complete and if $D$ is the Dirac operator, then $i^q D$ is essentially Krein-self-adjoint. In particular, if $\M$ is compact, then $i^q D$ is always essentially Krein-self-adjoint.
\end{proposition}

From all these properties, we can introduce the definition of a pseudo-Riemannian spectral triple:

\begin{definition}[\cite{Stro}]
A  pseudo-Riemannian Spectral Triple \mbox{$(\A,\H,D)$} is the data of:
\begin{itemize}
\item A Krein space $\mathcal{H}$,
\item A pre-$C^*$-algebra $\mathcal{A}$ with a faithful representation as bounded multiplicative operators on $\mathcal{H}$ and such that \mbox{$a^*=a^+$},
\item A Krein-self-adjoint operator $D$ on $\mathcal{H}$ such that all commutators $[D,a]$ are bounded for every $a\in\A$.
\end{itemize}
\end{definition}

In addition, we can also assume the existence of a fundamental symmetry $\J$ which commutes with all elements in $\A$. In this case, $\mathcal{A}$ becomes a subalgebra of $\B(\H)$ and the involution $a^*$ corresponds to the adjoint for the Hilbert space defined by $\J$. We can notice that in the commutative case, with this definition of pseudo-Riemannian spectral triple as given in \cite{Stro}, the operator $D$ corresponds to the usual Dirac operator times a factor $i^q$.

The compact resolvent condition is not present in this definition. This comes from the fact that, on a pseudo-Riemannian manifold, the Dirac operator $D$ is not elliptic. Its principal symbol satisfies the relation \mbox{$\sigma^D(\xi)^2 = c^2(\xi) = g(\xi,\xi)$} and so it is not invertible any more. In order to recover a similar condition, we define
\[\Delta_\J = \prt{1+ [D]_\J^2}^{\frac 12}\]
with $[D]_\J^2 = \frac{DD^*+D^*D}{2}$. This operator is elliptic of order $1$ since $\sigma^{\Delta_\J}(\xi)^2  = g^r(\xi,\xi)$, and is self-adjoint for the $\J$-product. The compact condition is then required on $\Delta_\J^{-1} = \prt{1+ [D]_\J^2}^{-\frac {1}{2}}$ and this is independent of the choice of the fundamental symmetry $\J$ \cite{Stro}.

\section{Temporal Lorentzian Spectral Triples}\label{TLSP}

Pseudo-Riemannian spectral triples are a good tool in order to deal with general pseudo-Riemannian signatures, but the definition in itself does not allow us to control the signature of the metric. In particular, it is not possible to highlight spectral triples corresponding to Lorentzian spaces, which contain interesting new properties like causality and are of great importance to physicists. We present here an adaptation of the axioms of pseudo-Riemannian spectral triples corresponding to noncompact globally hyperbolic Lorentzian manifolds. The definition of temporal Lorentzian spectral triples is presented as a working basis, since the set of axioms is not fixed and could evolve to include more properties.

We start by some considerations about the construction of the algebra. For a noncompact manifold, the dual algebra for Gel'fand theory is the nonunital algebra $C_0(\M)$ of continuous functions vanishing at infinity. The noncompactness of the manifold is mandatory in order to have a causal structure. Indeed, for a Lorentzian manifold without boundary, compactness implies the existence of points which are in the future of themselves \cite{Gallo,Tipler}, which is not physically realistic.

In order to deal with the noncompactness of the manifold, we will follow the construction in \cite{Gayral} for noncompact Riemannian spectral triples. We consider two pre-$C^*$-algebras $\A \subset \widetilde\A$ with a similar faithful representation on a Hilbert space $\H$. The algebra $\A$ is non-unital while the algebra $\widetilde\A$ is a preferred unitization of $\A$. Moreover, we require that $\A$ is an ideal of $\widetilde\A$ (note that a weaker axiom can be used by requiring that $\overline\A$ is an ideal of $\overline{\widetilde\A}$). Such a preferred unitization can be seen as a sub-algebra of the multiplier algebra of $\A$.

In the commutative case, those algebras can be constructed in the following way. Let us consider a noncompact Lorentzian spin manifold $\M$. We consider the Hilbert space $\H = L^2(\M,S)$ of square integrable spinor sections over $\M$ and the Dirac operator $D = -i(c \circ \nabla^S) = -i \gamma^{\mu} \nabla^S_\mu$. Then $\A \subset C^\infty_0(\M)$ and $\widetilde\A \subset C^\infty_b(\M)$ are some appropriate sub-algebras of the algebra of smooth functions vanishing at infinity and the algebra of smooth bounded functions in order to keep the property that $\forall a\in\widetilde\A$, $[D,a]$ is bounded. As an example, on Minkowski spacetime, one can take $\A$ to be the space of Schwartz functions and $\widetilde\A$ the space of smooth bounded functions with bounded derivatives \cite{Gayral}.

The main difficulty in defining a pseudo-Riemannian spectral triple corresponding to a Lorentzian manifold is to give a constraint on the signature such that only one timelike direction remains possible. In the construction we suggest here, such a constraint will take the form of an element representing a global time and defining by itself the needed fundamental symmetry in order to recover the self-adjointness of the Dirac operator. In the commutative case, we will suppose that the corresponding manifold respects the condition of global hyperbolicity.

From now, we suppose that $\M$ is a noncompact globally hyperbolic Lorentzian manifold admitting a spin structure. Since $\M$ is globally hyperbolic, $\M$ admits at least a Cauchy temporal function $\T$, which is a smooth time function (function strictly increasing along each future-directed causal curve) with past-directed timelike gradient everywhere and such that level sets are Cauchy surfaces \cite{BS03,BS04,BS06}. The Lorentzian metric admits a globally defined  orthogonal splitting 
\[g = -u \,d\T^2 + g_\T,\]
where $g_\T$ is a Riemannian metric on each level set of $\T$ and $u$ is a smooth positive function on $\M$ (we use here a signature of the type $(-,+,+,+,\dots)$). The function $\T$ is in general unbounded and thus does not belong to $\widetilde\A$ but its smooth inverse $\prt{1+ \T^2}^{-\frac{1}{2}}$ is a smooth bounded function.
 
At first, we consider a specific conformal transformation of the metric \mbox{$\tilde g = \frac 1u g$} and we suppose that the spin structure corresponds to the metric $\tilde g$. Such a conformal transformation induces an orthonormal splitting of the metric
 \[\tilde g = - d\T^2 + \tilde g_\T.\]
This splitting on the metric induces a splitting on the tangent (and the cotangent) bundle \mbox{$T\M = T\M_- \oplus T\M_+$}, where the subbundle $T\M_-$ has a one dimensional fiber generated by the gradient $\nabla\T$. So the temporal function $\T$ defines a spacelike reflection, with the associated Riemannian metric being \mbox{$\tilde g^r = d\T^2 + \tilde g_\T$}. Since $\nabla\T$ is a generating element of $T\M_-$ of norm one, $d\T$ is a generating element of $T\M_-^*$ of norm one. From the Proposition \ref{propslr}, we know that each spacelike reflection generates a fundamental symmetry $\J$ defined from the Clifford action $c$ on a local oriented orthonormal basis \mbox{$\set{e}$} of $T\M_-^{*}$ by $\J = ic(e)$. Since \mbox{$\set{d\T}$} is an orthonormal basis of $T\M_-^{*}$, we have that
\[\J = ic(d\T) = i\gamma^0\]
 is a fundamental symmetry, and if the manifold $\M$ is complete under the metric \mbox{$\tilde g^r = d\T^2 + \tilde g _\T$}, by the Proposition \ref{cpt1} $iD$ is Krein-self-adjoint for the Krein space defined by $\J$, which is equivalent to the fact that $\J D$ is a skew-self-adjoint operator in $\H$.
 
Then, we can recall that the Dirac operator respects the following known property \cite{Var}:
\[\forall f\in C^\infty(\M), \qquad [D,f] = -i\,c(df).\]
From that, we obtain that 
\[\J = - [D,\T]\] is a fundamental symmetry of the Krein space. The condition on the Krein-self-adjointness of $iD$ becomes a condition on the skew-self-adjointness of $[D,\T]D$ on the Hilbert space $\H$ and can be written as $D^* = -\J D \J = - [D,\T]D[D,\T]$.

Such a construction is satisfactory if we want information about the causal structure of the Lorentzian manifold, since the causality is unchanged by conformal transformation. However, in order to keep the metric aspect, one would like to extend the construction to more general conformal transformations. This can be done by relaxing the condition $\J^2 = 1$. Indeed, for a general metric $g = -u \,d\T^2 + g_\T$ and the corresponding spin structure, we can define a more general fundamental operator $J_u= - [D,\T] = i\gamma^0$ where $\gamma^0$ is now the curved Dirac matrix such that $\gamma^0\gamma^0 = g^{00} 1 = -u 1$. Since the fundamental symmetry is only modified by a Hermitian element $\J \rightarrow \J_u={u}^{\frac 12}\J$ with $u>0$, it still defines a Krein space where the operator $iD$ is Krein-self-adjoint, but it has lost its symmetric property and respects instead the more general condition $\J^2 = u1$. In order to fit our $C^*$-algebra formalism, we will restrict the conformal factors $u$ to elements in $\widetilde\A$, such that $\J^2$ can be seen as a positive (Hermitian invertible) element of $\widetilde\A$.

We now propose the following set of axioms for a Lorentzian spectral triple corresponding to a (possibly noncommutative) globally hyperbolic Lorentzian manifold.

\begin{definition}
A Temporal Lorentzian Spectral Triple \mbox{$(\A,\widetilde\A,\H,D,\T)$} is the data of:
\begin{itemize}
\item A Hilbert space $\H$,
\item A nonunital pre-$C^*$-algebra $\A$ with a faithful representation as bounded multiplicative operators on $\H$,
\item A preferred unitization $\widetilde\A$ of $\A$ with a similar faithful representation as bounded multiplicative operators on $\H$ and such that $\A$ is an ideal of $\widetilde\A$,
\item An unbounded self-adjoint operator $\T$ on $\H$ such that \mbox{$\prt{1+ \T^2}^{-\frac{1}{2}}\in \widetilde\A$},
\item An unbounded operator $D$ densely defined on $\mathcal{H}$ such that:
\begin{itemize}
\item all commutators $[D,a]$ are bounded for every $a\in\widetilde\A$,
\item $[D,\T]$ is a self-adjoint operator on $\H$ which commutes with every element in $\widetilde\A$ and such that $[D,\T]^2 = u \in \widetilde A$ with $u>0$,
\item \mbox{$[D,\T]D$} is a skew-self-adjoint operator on $\H$,
\item $\forall\,a\in\mathcal{A}$,   $a(1 + \scal{D}^2)^{-\frac 12}$ is compact, with\\ $\scal{D}^2 = -\frac{1}{2}\prt{D [D,\T] D [D,\T]+ [D,\T] D [D,\T] D}$.
\end{itemize}
\end{itemize}
\end{definition}

\begin{definition}
A temporal Lorentzian spectral triple \mbox{$(\A,\widetilde\A,\H,D,\T)$} is even if there exists a $\mathbb Z_2$-grading $\gamma$ such that $\gamma^*=\gamma$, $\gamma^2=1$, $[\gamma,a] = 0\ \forall a\in\widetilde\A$, \mbox{$\gamma \T = \T \gamma$} and \mbox{$\gamma D =- D \gamma $}.
\end{definition}

A temporal Lorentzian spectral triple is a pseudo-Riemannian spectral triple corresponding to a manifold with Lorentzian signature and admitting a global time, where the usual fundamental symmetry (usual in the sense of $\J^2=1$) is completely determined by the Dirac operator $D$ and the global time $\T$ by $\J=- u^{-\frac{1}{2}}[D,\T]$ with $u=[D,\T]^2$. We can notice that, when the algebra is noncommutative, the commutativity condition of $[D,\T]$ restricts the possibility of conformal transformations by requiring that $u$ belongs to the center of the algebra $Z(\widetilde\A)$. When the center is trivial, this reduces to the usual fundamental symmetry with $[D,\T]^2=1$ up to a constant scale. An example of construction of a temporal Lorentzian spectral triple in the noncommutative case will be given in Section \ref{NLMP}.

While the construction of a temporal Lorentzian spectral triple comes from the properties of Lorentzian manifolds with global hyperbolicity, the following proposition shows us that the Lorentzian characteristics are recovered by this choice of axioms.

\begin{proposition}\label{proprec}
Let us assume that a temporal Lorentzian spectral triple whose algebra is commutative \mbox{$(\A,\widetilde\A,\H,D,\T)$} corresponds to a pseudo-Riemannian spin manifold $(\M,g)$ in the usual way. Then the geometry is Lorentzian and the metric admits a global splitting.
\end{proposition}

\begin{proof}
We work with a spin structure on $(\M,g)$ such that the Dirac operator reads $D = -i \gamma^{\mu} \nabla^S_\mu$ where the gamma matrices $\gamma^{\mu} = c(dx^\mu)$ are chosen to be either Hermitian or anti-Hermitian depending on the signs of the metric and where the Clifford relations are $\gamma^{\mu} \gamma^{\nu} + \gamma^{\nu} \gamma^{\mu} = 2 g^{\mu\nu}$.

For arbitrary coordinates, we have $[D,\T]=- i \gamma^\mu \partial_\mu\T$, so the condition $[D,\T]^2 = u >0$ can only be respected if the gradient of $\T$ vanishes nowhere. Hence, $\T$ can be chosen as a first global coordinate $x^0=\T$ and we get $[D,\T]=- i \gamma^0$. Then $[D,\T]^2 = - g^{00}\, 1  > 0 \implies g^{00} < 0 $. We can notice that $\gamma^0$ is forced to be anti-Hermitian due to the self-adjointness of $[D,\T]$.\\
From the skew-self-adjoint condition $([D,\T]D)^*+[D,\T]D = 0$ and discarding a divergence term (in a similar way to \cite[Proposition 9.13]{Var}), we get
$$  
-i (\gamma^\mu)^*  \gamma^0 (-i \nabla^S_\mu)^*- i \gamma^0 \gamma^\mu  (-i \nabla^S_\mu) = -\left( (\gamma^\mu)^*  \gamma^0 +  \gamma^0 \gamma^\mu \right) \nabla^S_\mu = 0.
$$
Here we can discuss the Hermicity of the matrices $\gamma^\mu$ for $\mu>0$. If some matrix $\gamma^\mu$ has the same Hermicity that $\gamma^0$ (so is anti-Hermitian), we get $\gamma^\mu  \gamma^0 =  \gamma^0 \gamma^\mu$ which implies $ \gamma^\mu \gamma^0 = g^{\mu 0}\,1 $ from the Clifford relations, which is impossible since  $\gamma^0 \gamma^\mu$ cannot be a multiple of the identity matrix for $\mu \neq 0$. Hence all matrices $\gamma^\mu$ are Hermitian for $\mu > 0$ and they correspond to positive eigenvalues of the metric. Therefore, the geometry is Lorentzian. Moreover, we get $$\gamma^\mu  \gamma^0 +  \gamma^0 \gamma^\mu = 0 = g^{\mu0}\,1 = g^{0\mu}\,1$$ and the metric admits a global splitting.
\end{proof}

We must remark that the above proposition is not a reconstruction theorem for Lorentzian manifolds. The meaning is that, assuming that there exists a reconstruction theorem for pseudo-Riemannian spin manifolds, then the chosen class of fundamental symmetries $\J=- u^{-\frac{1}{2}}[D,\T]$ restricts the possibilities to those manifolds with Lorentzian signature. Also, the chosen class of fundamental symmetries does not cover all possible Lorentzian spin manifolds but only those which admit a global splitting.

In the rest of this section, we will show that the above definition of a temporal Lorentzian spectral triple is suitable for the definition of an extension of the algebra to unbounded elements. Indeed, one key element about Lorentzian noncommutative geometry is the set of causal functions, which are real-valued functions that do not decrease along every causal future-directed curve. In particular, the sets of time or temporal functions are subsets of the set of causal functions. Those causal functions are used in order to give information about causality and metric aspects (see \cite{F6,F3,F4,PZ,Mor,Bes}). In particular, the construction of a Lorentzian distance function requires some functions with a minimum growth rate along causal future-directed curves \mbox{$g( \nabla f, \nabla f ) \leq -1$} \cite{F4}, which is incompatible with the boundedness condition of the elements in $\widetilde\A$. So we would like to construct an algebra which could contain unbounded causal functions in order to keep information about time, Lorentzian distance and causality. However, in this case, the supremum norm (and a fortiori $L^p$ norms) cannot be used to define a Banach algebra. The needed unbounded elements are clearly in the $^*$-algebra $C^\infty(\M)$. Since smooth functions are locally integrable, we can use a structure defined for the space of locally integrable functions $L^1_{\text{loc}}(\M)$. Such a structure is provided by the theory of partial inner product spaces (PIP-spaces). We give here a few elements of this theory. A complete introduction, including the topological aspects, can be found in \cite{Antoine}.

\begin{definition}[\cite{Antoine}]
A linear compatibility relation on a vector space $V$ is a symmetric binary relation $f\# g$ which preserves linearity. For $S\subset V$, we can define the vector subspace $S^\#=\set{g\in V : g\# f, \forall f\in S} \subset V$ which respects the inclusion property \mbox{$S \subset \prt{S^\#}^\#$}. Vector subspaces such that \mbox{$S = \prt{S^\#}^\#$} are called assaying subspaces.
\end{definition}

The family of assaying subspaces forms a lattice  with the inclusion order, and the meet and join operations given by $S_1 \wedge S_2 =S_1 \cap S_2$ and  $S_1 \vee S_2 = \prt{S_1 + S_2}^{\#\#}$. Usually, it is enough to consider only an indexed sublattice \mbox{$\I = \set{S_r : r\in\I}$}  which covers $V$, with an involution defined on the index $I$ by $\prt{S_r}^\# = S_{\overline r}$.

\begin{definition}[\cite{Antoine}]
A partial inner product on $(V,\#)$ is a Hermitian form $\scal{\,\cdot\,,\,\cdot\,}$ defined exactly on compatible pairs of vectors $f\# g$. A partial inner product space (PIP-space) is a vector space $V$ equipped with a linear compatibility relation and a partial inner product. An indexed PIP-space is a PIP-space with a generating involutive indexed lattice of assaying subspaces.
\end{definition}

The space \mbox{$L^1_{\text{loc}}(\M)$} is a PIP-space with the compatibility relation given by
\[f\# g\quad\lequi\quad \int_\M \abs{fg} d\mu_g \ <\  \infty\]
and the partial inner product
\[\scal{f,g} = \int_\M f^*g \;d\mu_g .\]

We can define an involutive indexed lattice of assaying subspaces by using weight functions. If, for $r,r^{-1} \in L^2_{\text{loc}}(\M)$ with $r$ a.e.~positive Hermitian, we define the space $L^2(r)$ of measurable functions $f$ such that $fr^{-1}$ is square integrable, i.e.
\[L^2(r) = \set{f\in L^1_{\text{loc}}(\M) : \int_\M \abs{f}^2 r^{-2} \;d\mu_g \;<\; \infty},\]
then the family \mbox{$\I = \set{L^2(r)}_r$} respects $L^1_{\text{loc}} = \bigcup_{r} L^2(r)$ and is a generating lattice of assaying subspaces, with an involution defined by $\overline{r} = r^{-1}$ since \mbox{$L^2(r)^\# = L^2(r^{-1})$}.

Actually, we have a realization of this indexed PIP-space as a lattice of Hilbert spaces \mbox{$\set{\H_r}_r$} where each \mbox{$\H_r = L^2(r)$} is endowed with the positive definite Hermitian inner product
\[\scal{f,g}_r = \int_\M f^*g \; r^{-2} \;d\mu_g .\]
In the following, we will denote such an indexed PIP-space by $\bigcup_{r} \H_r$, with the so called center space $\H_0$ being the space of square integrable functions.

The space $C^\infty_b(\M)$ is a unital pre-$C^*$-algebra which acts as bounded multiplicative operators on $\H_0$, with the operator norm corresponding to the supremum norm \mbox{$\norm{f}_\infty = \sup_{x\in\M} \abs{f(x)}$}. We can make a similar modification to the supremum norm by introducing a smooth weight function $r$. We denote the algebra of bounded smooth functions by $\widetilde\A_0 = C^\infty_b(\M)$. We define the following lattice of spaces:
\[\widetilde\A_r = \set{ f\in C^\infty(\M) : \sup_{x\in\M} \abs{f(x)\;r^{-1}(x)} \;<\;\infty}\cdot\]

Except for the center algebra $\widetilde\A_0$, each $\widetilde\A_r $ is a vector space which does not have the structure of an algebra. Instead, those spaces have a structure of partial *-algebra, which means that the product $fg\in \widetilde\A_r$ is well defined only for a bilinear subset of $\widetilde\A_r \times \widetilde\A_r $. Actually, we have the grading property $fg\in \widetilde\A_r$ if $f\in \widetilde\A_s$ and $g\in \widetilde\A_t$, with $r=st$. Each $\widetilde\A_r $ is endowed with a norm $\norm{\,\cdot\,}_r = \norm{\,\cdot\ r^{-1}}_0$ , where  $\norm{\,\cdot\,}_0$ is the operator norm on $\widetilde\A_0$.
In the same way as the lattice of Hilbert spaces, we define the $^*$-algebra \mbox{$\bigcup_{r} \widetilde\A_r$} which obviously corresponds to the algebra of smooth functions $C^\infty(\M)$ but endowed with a gradation.

It is trivial that \mbox{$ \bigcup_{r} \widetilde\A_r$} acts on the PIP-space $ \bigcup_{r} \H_r$ as multiplicative operators. Moreover, we have that each $a\in A^M_s$ acts as a bounded operator \mbox{$a : \H_r \rightarrow \H_{rs}$} if we consider the weighted norm.

So we see that it is possible to consider $C^\infty(\M)$ as a graded algebra of operators on some PIP-space, with a set of norms defined on the gradation. However, we do not need the whole space $C^\infty(\M)$. Indeed, the needed causal functions giving metric information are only unbounded "in the direction of time", which means that we do not really need functions which are growing indefinitely on spacelike surfaces. So we propose the following idea: the set of weight functions can be restricted in such a way that the filtered algebra \mbox{$\bigcup_{r} \widetilde\A_r$} contains some functions which are growing indefinitely only along causal curves. A typical way is to restrict $C^\infty(\M)$ to the functions which satisfy a particular growth rate along causal curves. The final algebra can be constructed by the use of the global time function $\T$:
\begin{itemize}
\item $\widetilde\A_0 = C_b^\infty(\M)$ is the unital pre-$C^*$-algebra of smooth bounded functions on $\M$.
\item $\widetilde\A_n = \set{ f\in C^\infty(\M) : \sup_{x\in\M} \abs{f(x)\; \prt{1+\T(x)^2}^{\frac n2}} \;<\;\infty}, n\in\setN$.
\item \mbox{$\bigcup_{n\in\setZ} \widetilde\A_n$} is the unital filtered algebra of smooth  functions on $\M$ of polynomial growth along timelike curves which are bounded on each Cauchy surface $\T^{-1}(t)$, $t\in\setR$.
\end{itemize}

This algebra has a large number of interesting properties, easily derived from the definition:
\begin{itemize}
\item Each set $\widetilde\A_n$ is simply generated from $\widetilde\A_0$ by $f\in\widetilde\A_n$ if and only if $f\prt{1+\T(x)^2}^{\frac n2} \in\widetilde\A_0$.
\item Each set $\widetilde\A_n$ is endowed with a norm $\norm{\,\cdot\,}_n = \norm{\,\cdot\ \prt{1+\T(x)^2}^{\frac n2}}_0$ where $\norm{\,\cdot\,}_0$ is the norm on $\widetilde\A_0$.
\item We have the (descending) filtration properties $\widetilde\A_n \subset \widetilde\A_{n-1}\; \forall n\in\setZ$ and
\[ab\in\widetilde\A_{n} \ \lequi\ \exists \,m,l\in\setZ\ :n = m + l \ ,\ a\in\widetilde\A_m, \ b\in\widetilde\A_l.\]
\item $\T\in\bigcup_{n\in\setZ} \widetilde\A_n$, and more precisely $\T\in\widetilde\A_{-1}$.
\item All smooth causal functions with polynomial growth are included in the algebra $\bigcup_{n\in\setZ} \widetilde\A_n$.
\item All continuous a.e.~differentiable causal functions with polynomial growth are in the closure algebra \mbox{$ \bigcup_{n\in\setZ} \overline{\widetilde\A_n}$}, where the closure is taken in each subset with respect to each norm $\norm{\,\cdot\,}_n$.
\end{itemize}

The polynomial growth is just a particular choice. For example, one can define a similar filtered algebra with causal functions of exponential growth by taking a weight function like \mbox{$e^{\alpha \abs{\T}}$}, $\alpha\in\setR$. The choice of the center unitization is also arbitrary since $\widetilde\A_0$ can be replaced by another suitable unitization.

Now, we can remember that the pre-$C^*$-algebra $\A = C_0^\infty(\M)$, as well as its unitization $\widetilde\A=\widetilde\A_0 = C_b^\infty(\M)$, act as multiplicative bounded operators on the Hilbert space $\H = L^2(\M,S)$ of square integrable spinor sections over $\M$. From that, we can construct an indexed PIP-space $ \bigcup_{n\in\setZ} \H_n$ of spinor sections over $\M$ which are square integrable under the weighted norm, and on which $ \bigcup_{n\in\setZ} \widetilde\A_n$ acts as a family of bounded operators.

Such a space corresponds to a scale of Hilbert spaces generated by the self-adjoint operator \mbox{$\prt{1+\T^2}^{\frac 12} \geq 1$}. Indeed, \mbox{$\T : \text{Dom}\prt{\T}  \rightarrow \H_0$} is an unbounded self-adjoint operator acting on $\H = \H_0= L^2(\M,S)$ with domain \mbox{$\text{Dom}\prt{\T}\subset \H_0$}, and \mbox{$\prt{1+\T^2}^{\frac 12}$} is an unbounded positive self-adjoint operator with similar domain \mbox{$\text{Dom}\prt{\T} = \text{Dom}\prt{\prt{1+\T^2}^{\frac 12}}$}.

We can define the following Hilbert spaces:
\[\H_n = \bigcap_{k=0}^n \text{Dom}\prt{\prt{1+\T^2}^{\frac k2}} = \bigcap_{k=0}^n \text{Dom}(\T^k)\quad\forall n\in\setN.\]
The conjugate spaces \mbox{$\H_{\overline n} = \H_{-n}$}, $n\in\setN$ are just the topological duals of the spaces $\H_n$ (properly speaking, there are the identifications of the topological duals under the Hermitian inner product on $\H_0$).

Then we have the discrete scale of Hilbert spaces:
\[ \bigcap_{n\in\setZ} \H_n \subset \dots \subset \H_2 \subset \H_1 \subset \H_0 \subset \H_{-1} \subset \H_{- 2} \subset \dots \subset  \bigcup_{n\in\setZ} \H_n. \]

Each $\H_n$ is endowed with the positive definite Hermitian inner product
\[\scal{\psi,\phi}_n = \scal{\psi,\prt{1+\T^2}^{ n} \phi} =  \int_\M \psi^*\phi \, \prt{1+\T^2}^{ n}\,d\mu_g.\]

Then $ \bigcup_{n\in\setZ} \widetilde\A_n$ acts as a family of multiplicative operators on $ \bigcup_{n\in\setZ} \H_n$ with  \mbox{$a : \H_n \rightarrow \H_{n + m}$} bounded if \mbox{$a\in\widetilde\A_m$} and for each $n\in\setZ$.

This construction can be generalized to noncommutative spaces using temporal Lorentzian spectral triples. However, in order to conserve the properties of this construction, we need to impose some additional commutative conditions.

\vspace{0.2cm}

\begin{definition}
A temporal Lorentzian spectral triple \mbox{$(\A,\widetilde\A,\H,D,\T)$} has commutative time if the operators $\prt{1+\T^2}^{\frac 12}$ and \mbox{$[D,\prt{1+\T^2}^{\frac 12}]$} commute with all elements in $\widetilde\A$.\\
\end{definition}

\vspace{0.1cm}

\begin{proposition}
If a temporal Lorentzian spectral triple \mbox{$(\A,\widetilde\A,\H,D,\T)$} has commutative time, then:
\begin{itemize}
\item $\T$ generates a filtered algebra \mbox{$\bigcup_{n\in\setZ} \widetilde\A_n$} by \vspace{-0.2cm}
\[ \vspace{-0.2cm} \widetilde\A_0 = \widetilde\A,\qquad a\in \widetilde\A_{n+1}\quad\text{ if and only if }\quad  \prt{1+ \T^2}^{\frac{1}{2}} a \in \widetilde\A_{n},\] with a norm \mbox{$\norm{\,\cdot\,}_n = \norm{(1+\T^2)^{\frac n2}\ \cdot\,}$} defined on each $ \widetilde\A_n$.
\item $\bigcup_{n\in\setZ} \widetilde\A_n$ has a faithful representation as a family of multiplicative operators on the indexed PIP-space $\bigcup_{n\in\setZ} \H_n$ defined by \vspace{-0.2cm}
\[ \vspace{-0.2cm} \H_0=\H,\quad \quad\H_n = \bigcap_{k=0}^n \text{Dom}(\T^k)\quad\forall n > 0\]
\[\text{and }\  \H_{-n} =  \H_{\overline n} \ \text{ the conjugate dual of }\ \H_n\]
with \mbox{$a : \H_n \rightarrow \H_{n+m}$} bounded if \mbox{$a\in\widetilde\A_m$} and with a positive definite inner product  \mbox{$\scal{\,\cdot\,,\,\cdot\,}_{n}=\scal{\,\cdot\,,\prt{1+ \T^2}^{n}\,\cdot\,}$} defined on each $\H_n$.
\end{itemize}
Moreover, the norms \mbox{$\norm{a}_m = \norm{(1+T^2)^{\frac m2} a}$} and \mbox{$\norm{[D,a]}_m = \norm{(1+T^2)^{\frac m2} [D,a]}$} for $a\in \widetilde\A_m$ correspond to the operator norm on $\bigcup_{n\in\setZ} \H_n$.
\end{proposition}

\begin{proof}
We know that $\prt{1+ \T^2}^{\frac{1}{2}}$ commutes with all elements in $\widetilde\A$. This implies that $\prt{1+ \T^2}^{-\frac{1}{2}} \in \widetilde\A$ commutes with all elements in $\widetilde\A$ since $\forall a \in \widetilde\A$, $[\prt{1+ \T^2}^{-\frac{1}{2}},a] \prt{1+ \T^2}^{\frac{1}{2}} = 0$ and the codomain of $\prt{1+ \T^2}^{\frac{1}{2}}$ is $\H$. The filtration property \vspace{-0.1cm}
\[a\in\widetilde\A_m, \ b\in\widetilde\A_n\ \implies \ ab\in\widetilde\A_{m+n}\ \text{ and }\ ba\in\widetilde\A_{m+n}\]
comes from the fact that $\bigcup_{n\in\setZ} \widetilde\A_n$ is generated by multiple applications of $\prt{1+ \T^2}^{\frac{1}{2}}$ and $\prt{1+ \T^2}^{-\frac{1}{2}}$ on $\widetilde\A_0= \widetilde\A$, and the whole construction is similar to the commutative case.

The norm \mbox{$\norm{\,\cdot\,}_m = \norm{(1+T^2)^{\frac m2}\ \cdot\,}$} on $ \widetilde\A_m$ is the operator norm. Indeed, if we consider $a\in\widetilde\A_m$ as an operator \mbox{$a : \H_n \rightarrow \H_{n+m}$} for some $n\in\setZ$, then:
\[
\norm{a}_{\text{op}} = \!\!\!\!\sup_{\phi\in\H_n,\;\phi\neq0}\!\!\!\! \frac{\scal{a\,\phi , a\,\phi }_{n+m}}{ \scal{\phi , \phi}_n} =  \!\!\!\!\sup_{\phi\in\H_n,\;\phi\neq0}\!\!\!\! \frac{\scal{a\,\phi\ ,\  \prt{1+ \T^2}\!^{n+m} a\,\phi }}{ \scal{\phi\ ,\  \prt{1+ \T^2}^{n} \phi}}
\]
\[
 = \!\!\!\!\sup_{\phi\in\H_n,\;\phi\neq0}\!\!\!\! \frac{\scal{\prt{1+ \T^2}\!\!^{\frac m2}a\,\prt{1+ \T^2}\!\!^{\frac n2}\phi\ ,\  \prt{1+ \T^2}\!\!^{\frac m2} a\,\prt{1+ \T^2}\!\!^{\frac n2}\phi }}{ \scal{\prt{1+ \T^2}^{\frac n2} \phi \ ,\  \prt{1+ \T^2}^{\frac n2} \phi}}
 \]
\[
 = \!\!\!\!\sup_{\psi\in\H,\;\psi\neq0}\!\!\!\! \frac{\scal{\prt{1+ \T^2}\!\!^{\frac m2}a\,\psi\ ,\  \prt{1+ \T^2}\!\!^{\frac m2} a\,\psi }}{ \scal{\psi \ ,\   \psi}}
\]
\[
 = \norm{\prt{1+ \T^2}\!\!^{\frac m2}a} = \norm{a}_m
 \]
where we use the facts that $\prt{1+ \T^2}\!^{\frac 12}$ and $a$ commute with each other and that $\prt{1+ \T^2}\!^{\frac 12}$ is self-adjoint. The result is clearly independent of the chosen $\H_n$.

The condition that \mbox{$[D,\prt{1+\T^2}^{\frac 12}]$} commutes with all elements in $\widetilde\A$ is equivalent to the fact that, for all $a\in\widetilde\A$, \mbox{$[D,a]$} commutes with $\prt{1+\T^2}^{\frac 12}$. Indeed, by using the fact that $\prt{1+\T^2}^{\frac 12}$ commutes with $a$:
\begin{eqnarray*}
 &&[D,\prt{1+\T^2}^{\frac 12}] \;a - a \;[D,\prt{1+\T^2}^{\frac 12}]\\
&=& D\prt{1+\T^2}^{\frac 12} a - \prt{1+\T^2}^{\frac 12} D a - a D \prt{1+\T^2}^{\frac 12} + a \prt{1+\T^2}^{\frac 12}D\\
&=& Da \prt{1+\T^2}^{\frac 12}  - \prt{1+\T^2}^{\frac 12} D a - a D \prt{1+\T^2}^{\frac 12} + \prt{1+\T^2}^{\frac 12} aD\\
&=& [D,a] \prt{1+\T^2}^{\frac 12} - \prt{1+\T^2}^{\frac 12} [D,a].
\end{eqnarray*}
Then, by a reasoning similar to above, the operator norm of $[D,a]$ is independent of the chosen $\H_n$.
\end{proof}

\section{The Moyal--Minkowski Spacetime}\label{NLMP}

Commutative temporal Lorentzian spectral triples can easily be constructed by using the construction presented in the previous section from any suitable globally hyperbolic Lorentzian manifold. One can wonder however if noncommutative examples can explicitly be constructed. In this section we answer the problem by showing that the Minkowski spacetime endowed with a Moyal product is a noncomutative temporal Lorentzian spectral triple.

The construction of spectral triples based on Euclidean Moyal planes has already been studied extensively in \cite{Gayral}. The algebra $\A$ is chosen as the space $\Sw(\setR^n)$ of Schwartz functions, which are smooth functions $f\in C_0^\infty (\setR^n)$ rapidly vanishing at infinity together with all derivatives. The product on $\Sw(\setR^n)$ is defined by the Moyal product:
\[(f \star h) (x) = \frac{1}{(2\pi)^n}\int \int f(x-\frac 12 \Theta u)\ h(x+v)\ e^{-iu \cdot v} \;d^nu\;d^nv\]
where $\Theta$ is a real skewsymmetric $n \times n$ constant matrix. This product can be extended to the space $L^2(\setR^n)$ \cite{Gayral,VG88} and there is a representation of $\A$ on the space of square integrable spinors $L^2(\setR^n) \otimes \setC^{2^{\lfloor{n/2}\rfloor}}$ as a multiple of the left regular action
\[\pi(f) \psi = (\L(f) \otimes 1) \psi = f \star \psi.\] 
$\A$ is a nonunital involutive Fréchet algebra, and if equipped with the operator norm
\[\norm{a}_{\text{op}} = \sup_{b\in L^2(\setR^n), b \neq 0} \frac{\norm{a\star b}_2}{\norm{b}_2},\]
 $\A$ becomes a nonunital pre-$C^*$-algebra. The pre-$C^*$-algebra $\A=\set{\Sw(\setR^n),\star}$, the Hilbert space $\H = L^2(\setR^n) \otimes \setC^{2^{\lfloor{n/2}\rfloor}}$ and the Dirac operator $D = -i \partial_\mu \otimes \gamma^\mu$ define a Euclidean noncompact spectral triple of spectral dimension $n$.

We construct the Moyal--Minkowski spacetime in a similar way. We work with the space $\setR^{1,n-1}$, and the Moyal $\star$ product is defined on the space of Schwartz functions $\Sw = \Sw(\setR^{1,n-1})$ and extended to $L^2(\setR^{1,n-1})$. We define the nonunital pre-$C^*$-algebra $\A=\prt{\Sw,\star}$ acting on the Hilbert space $\H = L^2(\setR^{1,n-1}) \otimes \setC^{2^{\lfloor{n/2}\rfloor}}$. We must choose a suitable unitization $\widetilde\A$ for the temporal Lorentzian spectral triple. A first choice could be the (left and right) Moyal multiplier algebra:
\[ M(\A) = \set{ T \in \Sw' : T \star h  \in \Sw\text{ and } h \star T  \in \Sw \ \text{ for all } h  \in \Sw}\]
where $\Sw' = \Sw'(\setR^{1,n-1})$ is the dual space of tempered distributions. Coordinates functions $x^0, x^1, \dots x^{n-1}$ belong to this space, and respect the following commutation relations with the star product:
\[ [x^\mu,x^\nu] = x^\mu \star x^\nu -  x^\nu \star x^\mu = i\, \Theta^{\mu\nu} 1.\] However, this unitization cannot be chosen since some elements in $M(\A)$ are unbounded. Instead we choose the sub-algebra $\widetilde\A = \prt{\B,\star} \subset M(\A)$ of smooth functions which are bounded together with all derivatives. $\widetilde\A = \prt{\B,\star}$ is a unital Fréchet pre-$C^*$-algebra that contains $\A$ as an ideal \cite{Gayral}. To complete the  temporal Lorentzian spectral triple, we choose the temporal element $\T = x^0$ as the first coordinate function and the usual flat Dirac operator $D = -i \partial_\mu \otimes \gamma^\mu$.

\begin{proposition}
The triple $(\A,\widetilde\A,\H,D,\T) =$
 \[\prt{\prt{\Sw(\setR^{1,n-1}),\star},\prt{\B(\setR^{1,n-1}),\star},L^2(\setR^{1,n-1}) \otimes \setC^{2^{\lfloor{n/2}\rfloor}},-i \partial_\mu \otimes \gamma^\mu, x^0}\]
respects the chosen axioms for a temporal Lorentzian spectral triple.
\end{proposition}

\begin{proof}
\begin{itemize}\itemsep=6pt
\item $\mathcal{H} = L^2(\setR^{1,n-1}) \otimes \setC^{2^{\lfloor{n/2}\rfloor}}$ is a Hilbert space  with positive definite inner product defined by the usual inner product on spinors $\scal{\psi,\phi} = \int \psi^* \phi\, d\mu_g$.

\item $\A = \prt{\Sw(\setR^{1,n-1}),\star} $ is a  nonunital pre-$C^*$-algebra with a representation $\pi(f) = \L(f) \otimes 1$ as bounded multiplicative operators on $\mathcal{H}$  with the norm \mbox{$\norm{\,\cdot\,} = \norm{\,\cdot\,}_{\text{op}}$}. 

\item $\widetilde\A = \prt{\B(\setR^{1,n-1}),\star}$ is a preferred unitization of $\mathcal{A}$ with the same representation $\pi$ on $\mathcal{H}$ and with the same operator norm \mbox{$\norm{\,\cdot\,}$}.

\item $\T = x^0$ acts as an unbounded self-adjoint operator on $\H$. $\prt{1+ \T^2}^{-\frac{1}{2}} = \frac{1}{\sqrt{1+(x^0)^2}}$ is bounded on all $\setR^{1,n-1}$ with bounded derivatives, so it belongs to $\widetilde\A$.

\item $D=-i \partial_\mu \otimes \gamma^\mu$ is clearly an unbounded operator densely defined on $\mathcal{H}$. The property that $[D,a]$ acts as a multiplicative operator $D(a) = -ic(da)$ is still valid in the nonunital case \cite{Gayral}, so $[D,a]$ is bounded for every $a\in\widetilde\A$. 

\item From $[D,x^0] = -i c(dx^0) = -i \gamma^0$, the Krein space conditions can easily be checked:
\begin{itemize}
\item (Hermicity)\quad $[D,x^0]^* = \prt{-i \gamma^0}^* = i \prt{\gamma^0}^* = -i \gamma^0 = [D,x^0]$,
\item (Commutativity)\quad $[D,x^0] = -i \gamma^0$ commutes with every element in $\widetilde\A$ since $-i\in Z(\widetilde\A)$,
\item (symmetry)\quad $[D,x^0]^2 = -g^{00} = 1$,
\item (skew-self-adjointness)\quad $\prt{[D,x^0]D}^* = -[D,x^0]D$
\[\lequi -\prt{\gamma^\mu}^* \gamma^0 = \gamma^0 \gamma^\mu \quad \text{for }\mu \geq 0\]
\[ \lequi  \gamma^0 \gamma^0 = \gamma^0 \gamma^0 \quad \text{and } \gamma^\mu \gamma^0 = -\gamma^0 \gamma^\mu \quad\text{for } \mu > 0.\]
\end{itemize}
\end{itemize}
\end{proof}

If the dimension $n$ is even, then this temporal Lorentzian spectral triple is even by taking the chirality element \mbox{$\gamma = (-i)^{\frac{n}{2}+1}\gamma^0\dots\gamma^{n-1}$}.

\begin{proposition}
The temporal Lorentzian spectral triple on the Moyal--Minkowski spacetime has commutative time if and only if the matrix defining the Moyal product is degenerate with
\[\Theta^{0\mu} = \Theta^{\mu 0} = 0 \qquad \text{for all }\ \mu=0,\dots,n-1.\]
\end{proposition}

\begin{proof}
The degeneration of the matrix is equivalent to the fact that $\T=x^0$ is in the center of $M(\A)$, which implies by expansion that $\sqrt{1+(x^0)^2}$ commutes with all elements in $\widetilde\A\subset M(\A)$. Moreover, the property that \mbox{$[D, \sqrt{1+(x^0)^2}]$} commutes with all elements $a\in \widetilde\A$ is equivalent to the fact that $[D,a]$ commutes with $ \sqrt{1+(x^0)^2}$ which is true since $[D,a]$ acts as a simple multiplicative operator.
\end{proof}

We can notice that, when the Moyal--Minkowski spacetime has commutative time, this allows us to construct temporal Lorentzian spectral triples corresponding to a modification of the metric by a scale factor $u(t)$ with $[D,x^0]^2=u(t)$, since $u(t)$ belongs to the center of the algebra.

\section{A Lorentzian Distance Formula for even Temporal Lorentzian Spectral Triples with commutative time}\label{distsec}

We have proposed a definition of temporal Lorentzian spectral triples based on a time element $\T$ and shown that such a time element, under some commutative condition, allows us to extend the algebra to unbounded elements. In this last section, we show that such an extension can be used to define a Lorentzian distance formula for noncommutative spacetimes.

We first start by reviewing the results on the Lorentzian distance formula in the commutative case. We recall that a Lorentzian distance on a Lorentzian manifold $\M$ is a function $d : \M \times \M \rightarrow [0,+\infty) \cup \set{+\infty}$ defined by \vspace{0.2cm}
\[ \vspace{0.2cm}
d(p,q) =
\begin{cases}
\quad \sup \left\{ l(\gamma) : 
\begin{array}{c}
\gamma \text{ future directed causal}\\
\text{piecewise smooth curve}\\
\text{with } \gamma(0)=p,\ \gamma(1)=q
\end{array}\right\} \ &\text{if } p \preceq q\\ 
\quad 0 &\text{if } p \npreceq q
\end{cases}
 \]
 where $l(\gamma) = \int \sqrt{ -g_{\gamma(t)}( \dot\gamma(t) ,  \dot\gamma(t) )}\ dt$ is the length of the curve and $\preceq$ denotes the causal relation between points.
 
If $\M$ is globally hyperbolic, this function is finite and continuous, and respects the following properties \cite{Beem}:
\begin{enumerate}
\item $d(p,p) = 0$,
\item $d(p,q) \geq 0$ for all $p,q\in\M$,
\item If $d(p,q) > 0$, then $d(q,p) = 0$,
\item If $d(p,q) > 0$ and $d(q,r) > 0$, then $d(p,r) \geq d(p,q) + d(q,r)$.
\end{enumerate}

In the same way as Connes' distance formula gives an algebraic formulation of the Riemannian distance, there exists an algebraic formulation for the Lorentzian distance following the results in \cite{F6,F3}:

\begin{definition}\label{defnewdist}[\cite{F6}]
Let us consider a globally hyperbolic Lorentzian spin manifold $\M$ with even dimension. For every two points $p,q\in\M$, we define:
\[
\widetilde d(p,q) := \inf_{f \in C^\infty(\M,\setR)}\set{ \max\set{0,f(q)-f(p)} :
\forall \phi \in \H, \scal{\phi,\J([D,f]+ i\gamma) \phi } \leq 0} 
\]
where $D$ is the Dirac operator, $\H$ is the Hilbert space of square integrable spinor sections over $\M$, $\gamma$ is the chirality element and $\J=i\gamma^0$ (curved gamma matrix).
\end{definition}

\begin{proposition}[\cite{F6}]
The function $\widetilde d(p,q)$ respects all the properties of a Lorentzian distance and, for every $p,q\in \M$, we have $d(p,q) \leq \widetilde d(p,q)$.
\end{proposition}

It is still unknown if there is an equality between this algebraic formula and the usual Lorentzian distance formula for every globally hyperbolic Lorentzian spin manifold. However, as it is shown in \cite{F6}, the equality can be obtained for manifolds whose Lorentzian distance function can be approximate by suitable smooth functions. In particular, on the Minkowski spacetime, this formula corresponds exactly to the usual Lorentzian distance. The proof of the equality for more general manifolds is currently a work in progress. We can conjecture that this is valid for a large class of globally hyperbolic Lorentzian spin manifolds.

We want to show that it is possible to adapt this algebraic formula to commutative temporal Lorentzian spectral triples with a possible generalization to noncommutative spaces. The difficulty comes from the constraint
\[
\forall \phi \in \H, \scal{\phi,\J([D,f]+ i\gamma) \phi } \leq 0
\]
which is equivalent, in the commutative case, to \cite{F6}
\[
\sup g( \nabla f, \nabla f ) \leq -1  \quad\text{and}\quad \nabla f \text{ is past-directed}
\]
which implies that the needed functions are unbounded along timelike curves, which is not possible in $\widetilde\A$. However, some functions of the filtered algebra \mbox{$\bigcup_{n\in\setZ} \widetilde\A_n$} fit our goal, and we can propose the following definition:

\begin{definition}\label{defnewdisttlost}
Let us consider a commutative even temporal Lorentzian spectral triple $(\A,\widetilde\A,\H,D,\T)$ with $\setZ_2$-grading $\gamma$ constructed from a globally hyperbolic Lorentzian spin manifold $\M$. For every two points $p,q\in\M$, we define:
\[
\widetilde d(p,q) := \inf_{f \in \bigcup \widetilde\A_n} \set{ \max\set{0,f(q)-f(p)} :
\forall \phi \in \H, \scal{\phi,[D,\T]([D,f]+ i\gamma) \phi } \geq 0} 
\]
where $\bigcup \widetilde\A_n$ is the filtered algebra generated by the time element $\T$.
\end{definition}

This definition is equivalent to the definition \ref{defnewdist}. Indeed, since the Lorentzian distance between two points is a local quantity, there is no influence on imposing conditions on the behaviour of the test functions at infinity (as long as their respect the condition $\sup g( \nabla f, \nabla f ) \leq -1$), so smooth functions with polynomial growth bounded on Cauchy surfaces are sufficient and lead to the same formula.

In order to generalize this definition to noncommutative spaces, we need to introduce a distance between states instead of a distance between points. We recall that states on $\A$ are positive linear functionals (automatically continuous) of norm one. The space of states is denoted by $S(\A)$. It is a closed convex set (for the weak-$^*$ topology), and extremal points are called pure states, with the set of pure states denoted by $P(\A)$. In the commutative case, those pure states are the characters of the algebra (non-zero *-homomorphisms). From Gel'fand--Naimark theory (see e.g. \cite{Var}), the pure states on $\A=C^\infty_0(\M)$ correspond to the points of the manifold $\M$ by the relation $\chi \sim p$ if and only if, $\forall f\in\A$, $\chi(f)=f(p)$ , for $\chi\in P(\A)$ and $p\in\M$. Those states can be extended uniquely to $\bigcup \widetilde\A_n \subset C^\infty(\M)$ by requiring that $\forall f\in\bigcup \widetilde\A_n$, $\chi(f)=f(p)$, and we have the following definition:

\begin{definition}\label{defnewdisttloststates}
Let us consider a commutative even temporal Lorentzian spectral triple $(\A,\widetilde\A,\H,D,\T)$ with $\setZ_2$-grading $\gamma$. For every two pure states $\chi,\xi\in P(\A)$, we define:
\[
\widetilde d(\chi,\xi) := \inf_{f \in \bigcup \widetilde\A_n} \set{ \max\set{0,\xi(f)-\chi(f)} :
\forall \phi \in \H, \scal{\phi,[D,\T]([D,f]+ i\gamma) \phi } \geq 0} 
\]
where $\bigcup \widetilde\A_n$ is the filtered algebra generated by the time element $\T$ and where $\chi$ and $\xi$ are considered as their unique extension to $\bigcup \widetilde\A_n$.
\end{definition}

This definition cannot be generalized to noncommutative spaces in a straight way. We need to study the way to extend the notion of states in the noncommutative case without using the notion of points, since such a notion becomes meaningless. First, we need to notice that, when the algebra is noncommutative, the space $P(\A)$ and $P(\widetilde\A)$ can be strongly different since the states on $\A$ may not necessarily be extended in a unique way to $\widetilde\A$. This will impose on us to only define a Lorentzian distance on $P(\widetilde\A)$. In order to extend the definition of the states to $\bigcup \widetilde\A_n$, we need the following lemma:

\begin{lemma}\label{lemmasubchar}
Let $\widetilde\A$ be a unital noncommutative pre-$C^*$-algebra and let $\chi \in P(\widetilde\A)$ be a pure state. If $a\in Z(\widetilde\A)$, then $\forall b \in \widetilde\A$ we have $\chi(ab)=\chi(a)\chi(b)$.
\end{lemma}

\begin{proof}
At first, let us suppose that $b\in \widetilde\A$ is normal, i.e. $bb^*=b^*b$. We can define the unital pre-$C^*$-algebra $A \subset \widetilde\A$ generated by the elements $\set{a,a^*,b,b^*,1}$. Since $a$ and $a^*$ commute with everything and $b$ is normal, $A$ is a commutative pre-$C^*$-algebra. We have that $\chi\in S(A)$, since $\chi$ is a positive linear functional on $A$ such that $\chi(1)=1$. We can see that $\chi$ is a pure state on $A$. Indeed, if we suppose that $\chi = \lambda \chi_1 + (1-\lambda) \chi_2$ for two states $\chi_1,\chi_2 \in S(A)$ with $0<\lambda<1$, then $\chi_1$ can be extended to $\widetilde\A$ as a state with $\abs{\lambda\chi_1} \leq \abs{\chi}$ by the Hahn--Banach theorem \cite{BraRo} as well as $\chi_2$ by setting $\chi_2 = \frac{1}{1-\lambda} \prt{\chi - \lambda \chi_1}$, which implies that $\chi_1=\chi_2$ since $\chi$ is a pure states on $\widetilde\A$. Since $A$ is commutative and $\chi$ is a pure state, it is a character and respects the *-homomorphism property on $A$, so $\chi(ab)=\chi(a)\chi(b)$.

Let us suppose now that $b$ is not normal. Then it can be written as $b=b_1+b_2$ where $b_1$ is Hermitian and $b_2$ anti-Hermitian, so $b_1$ and $b_2$ are normal. Then we have 
\[\chi(ab)=\chi(ab_1) + \chi(ab_2)= \chi(a)\chi(b_1) + \chi(a)\chi(b_2) =\chi(a)\chi(b).\]
\end{proof}

\begin{proposition}
Let us consider a temporal Lorentzian spectral triple with commutative time $(\A,\widetilde\A,\H,D,\T)$ . Then for every pure state $\chi\in P(\widetilde\A)$ such that $\chi(\prt{1+ \T^2}^{-\frac{1}{2}}) \neq 0$, $\chi$ admits a unique extension to the filtered algebra $\bigcup_{n\in\setZ} \widetilde\A_n$ defined by:
\[
\forall a \in \bigcup_{n\in\setZ} \widetilde\A_n,\quad \chi(a) =  \chi(\prt{1+ \T^2}^{-\frac{1}{2}}) ^{-n} \;\chi(a_0)
\]
\[
\text{ if }\quad a=\prt{1+ \T^2}^{\frac{n}{2}} a_0 \quad\text{ with }\quad a_0 \in\widetilde\A. 
\]
\end{proposition}

\begin{proof}
Let us recall that, for every $a\in \bigcup_{n\in\setZ} \widetilde\A_n$, there exist $a_0\in\widetilde\A$ and $n$ such that $a=\prt{1+ \T^2}^{\frac{n}{2}} a_0$. Moreover, since $\prt{1+ \T^2}^{-\frac 12} \in \widetilde\A$, for every $a,b\in \bigcup_{n\in\setZ}$, there exist $a_0,b_0\in\widetilde\A$ and a common $n$ such that $a=\prt{1+ \T^2}^{\frac{n}{2}} a_0$ and $b=\prt{1+ \T^2}^{\frac{n}{2}} b_0$. It follows that the extension of $\chi$ to $\bigcup_{n\in\setZ} \widetilde\A_n$ is linear. We have trivially that $\chi$ is positive with $\chi(1)=1$.

In order to show that this extension if well defined, we must prove that the definition of $\chi(a)$ is independent of the chosen $a_0$. Let us suppose that $a=\prt{1+ \T^2}^{\frac{n}{2}} a_0 =\prt{1+ \T^2}^{\frac{m}{2}}  a_0^\prime$ with $a_0,a_0^\prime\in\widetilde\A$. We have $a_0 =\prt{1+ \T^2}^{\frac{m-n}{2}}  a_0^\prime$. Since the temporal Lorentzian spectral triple has commutative time, $\prt{1+ \T^2}^{-\frac{1}{2}}$ is in the center of the algebra, and by the Lemma \ref{lemmasubchar} we have 
\[
\chi(a_0) =\chi(\prt{1+ \T^2}^{\frac{m-n}{2}}) \chi(a_0^\prime) = \chi(\prt{1+ \T^2}^{-\frac{1}{2}}) ^{n-m} \chi(a_0^\prime).
\] 
Therefore, 
\begin{eqnarray*}
\chi(a) &=&  \chi(\prt{1+ \T^2}^{-\frac{1}{2}}) ^{-n} \chi(a_0)\\
 &=& \chi(\prt{1+ \T^2}^{-\frac{1}{2}}) ^{-n} \chi(\prt{1+ \T^2}^{-\frac{1}{2}}) ^{n-m} \chi(a_0^\prime)\\
&=&  \chi(\prt{1+ \T^2}^{-\frac{1}{2}}) ^{-m} \chi(a_0^\prime)
\end{eqnarray*}
and the definition of the extension of $\chi$ is independent of the chosen $a_0 \in\widetilde\A$.
\end{proof}

We can notice that, in the commutative case, pure states $\chi\in P(\widetilde\A)$ such that $\chi(\prt{1+ \T^2}^{-\frac{1}{2}}) = 0$ do not correspond to pure states in $P(\A)$ since they do not correspond to an evaluation map at any point. They are states added by the unitization process (corresponding to a compactification of the manifold) and should be ignored. So it looks natural to ignore them even in the noncommutative case.

We can also see that, even if the requirement to have a commutative time could seem sometime restrictive, it is a necessary condition in order to guarantee that the extension of the states to unbounded elements is unique, and so to have a well defined noncommutative generalization of a Lorentzian distance formula.

At the end, we have the following definition for a Lorentzian distance formula valid on both commutative and noncommutative spaces:

\begin{definition}\label{defnewdisttlostsnc}
Let us consider an even temporal Lorentzian spectral triple $(\A,\widetilde\A,\H,D,\T)$ with commutative time and with $\setZ_2$-grading $\gamma$. For every two pure states $\chi,\xi\in P(\widetilde\A)$  such that $\chi(\prt{1+ \T^2}^{-\frac{1}{2}}) \neq 0$ and $\xi(\prt{1+ \T^2}^{-\frac{1}{2}}) \neq 0$, we define:
\[
\widetilde d(\chi,\xi) := \inf_{a \in \bigcup \widetilde\A_n} \set{ \max\set{0,\xi(a)-\chi(a)} :
\forall \phi \in \H, \scal{\phi,[D,\T]([D,a]+ i\gamma) \phi } \geq 0} 
\]
where $\bigcup \widetilde\A_n$ is the filtered algebra generated by the time element $\T$ and where $\chi$ and $\xi$ are considered as their unique extension to $\bigcup \widetilde\A_n$ defined by:
\[
\chi(a) =  \chi(\prt{1+ \T^2}^{-\frac{1}{2}}) ^{-n} \;\chi(a_0)
\quad \text{ and } \quad \xi(a) =  \xi(\prt{1+ \T^2}^{-\frac{1}{2}}) ^{-n} \;\xi(a_0)
\]
\[
\text{ if }\quad a=\prt{1+ \T^2}^{\frac{n}{2}} a_0 \quad\text{ with }\quad a_0 \in\widetilde\A. 
\]
\end{definition}

\section{Conclusions}

In this paper, we have studied the problem of defining spectral triples corresponding to manifolds with Lorentzian signature. We have suggested some axioms to construct temporal Lorentzian spectral triples, which are spectral triples corresponding, in the commutative case, to globally hyperbolic Lorentzian manifolds. Those axioms are based on the existence of an element representing a global time and which is sufficient to define a fundamental symmetry needed to recover the self-adjointness of the Dirac operator. This can be shown as a kind of "3+1" decomposition of a possibly noncommutative Lorentzian manifold. An explicit construction of a noncommutative example is given in Section \ref{NLMP} with a Moyal--Minkowski spacetime. In the commutative case, the Proposition \ref{proprec} guarantees that, even with an abstract definition of temporal Lorentzian spectral triples, the corresponding metric has a Lorentzian signature and a split nature.

We have shown that, under the requirement that time is commutative, those axioms allow us to extend the usual $C^*$-algebra formalism to unbounded elements by using a technique of filtration and partial inner product spaces (PIP-spaces). This extension is a needed tool in order to extend the definition of a Lorentzian distance formula  \cite{F6} to noncommutative spaces. The noncommutative Lorentzian distance (Definition \ref{defnewdisttlostsnc}) is defined at the level of extensions of pure states to unbounded elements which are well-defined thanks to the filtration properties. At this time, the correspondence between such a formula and the usual Lorentzian distance in the commutative case is only completely proven for Minkowski spacetime, the correspondence resulting in an inequality in the general case. The possibility to define such distance on the space of pure states without requiring a commutative time is still an open question, as well as the complete proof of the correspondence with the usual Lorentzian distance. We must remark that, for more general manifolds and possibly noncommutative time, a definition of causality on the space of states can always be defined as it is shown in \cite{F6}.

The axioms of temporal Lorentzian spectral triples are subject to further development. In particular, the question of the smoothness of the Dirac operator (or more precisely of its elliptic modification) could be taken into account. It could also be interesting to wonder which temporal elements $\T$ generate a similar Lorentzian manifold, and in this case to define a set of axioms for Lorentzian spectral triples which should be independent of the choice of a temporal element. The theory of Lorentzian noncommutative geometry is clearly still in its early stages, and a lot of work should be done in order to have a complete set of axioms for Lorentzian spectral triples which could give rise to a possible reconstruction theorem.

\section*{Acknowledgement}

This work was supported by a grant from the John Templeton Foundation.


\begin{thebibliography}{99}

\bibitem{C94}
A. Connes, {\it Noncommutative Geometry}, Academic Press, \mbox{San Diego}, 1994.

\bibitem{MC08}
A. Connes and M. Marcolli, {\it Noncommutative Geometry, Quantum Fields and Motives}, American Mathematical Society, Colloquium Publications Vol. 55, Providence, 2008.

\bibitem{Var}
J. M. Gracia-Bondía, J. C. Várilly and H. Figueroa, {\it Elements of Noncommutative Geometry}, Birkhäuser, Boston, 2001.

\bibitem{C96}
A. Connes, {\it Gravity coupled with matter and the foundations of non-commutative geometry}, Comm. Math. Phys. {\bf 182} (1996), \mbox{155--176}, hep-th/9603053.

\bibitem{C08}
A. Connes, {\it On the spectral characterization of manifolds}, J. Noncommutative Geom. {\bf 7} Issue 1 (2013), 1--82, arXiv:0810.2088.

\bibitem{MC207}
A. H. Chamseddine, A. Connes and M. Marcolli, {\it Gravity and the standard model with neutrino mixing}, Adv. Theor. Math. Phys. 11 (2007), 991--108, hep-th/0610241.

\bibitem{Haw}
E. Hawkins, {\it Hamiltonian Gravity and Noncommutative Geometry}, Comm. Math. Phys. {\bf 187} (1997), 471--489, gr-qc/9605068.

\bibitem{Kopf98}
T. Kopf, {\it Spectral Geometry and Causality}, Int. J. Mod. Phys. A {\bf 13} (1998), 2693--2708, gr-qc/9609050.

\bibitem{Kopf00}
T. Kopf, {\it Spectral geometry of spacetime}, Int. J. Mod. Phys. B {\bf 14} (2000), 2359--2366, hep-th/0005260.

\bibitem{Kopf01}
T. Kopf and M. Paschke, {\it Spectral quadruple}, Mod. Phys. Lett. A {\bf 16} 4-6 (2001), 291--298, math-ph/0105006.

\bibitem{Kopf02}
T. Kopf and M. Paschke, {\it A spectral quadruple for de Sitter space}, J. Math. Phys. {\bf 43} (2002), 818--846, math-ph/0012012.

\bibitem{Stro}
A. Strohmaier, {\it  On noncommutative and pseudo-Riemannian geometry}, J. Geom. Phys. {\bf 56} (2006), 175--195, math-ph/0110001.

\bibitem{Suij}
W. D. van Suijlekom, {\it The noncommutative Lorentzian cylinder as an isospectral deformation}, J. Math. Phys. {\bf 45} (2004), 537--556, math-ph/0310009.

\bibitem{PasV}
M. Paschke and R. Verch, {\it Local covariant quantum field theory over spectral geometries},	Class. Quant. Grav. {\bf 21} (2004), 5299--5316, gr-qc/0405057.

\bibitem{Pas}
M. Paschke and A. Sitarz, {\it Equivariant Lorentzian spectral triples} (preprint), math-ph/0611029.

\bibitem{Rennie12}
K. van den Dungen, M. Paschke and A. Rennie, {\it Pseudo-Riemannian spectral triples and the harmonic oscillator}, J. Geom. Phys. {\bf 73} (2013), 37--55, arXiv:1207.2112.

\bibitem{F6}
N. Franco and M. Eckstein, {\it An algebraic formulation of causality for noncommutative geometry}, Class.~Quantum Grav. {\bf 30} (2013), 135007, arXiv:1212.5171.

\bibitem{F3}
N. Franco, {\it Global Eikonal Condition for Lorentzian Distance Function in Noncommutative Geometry}, SIGMA {\bf 6} (2010), 064, arXiv:1003.5651.

 \bibitem{BaumG} 
H. Baum, {\it Spin-Strukturen und Dirac-Operatoren über pseudoriemannschen Mannigfaltigkeiten}, Teubner-Text zur Mathematik {\bf 41}, Teubner, 1981. 

 \bibitem{BaumE} 
H. Baum, {\it A remark on the spectrum of the Dirac operator on a pseudo-Riemannian spin manifolds}, SFB288-preprint {\bf 136} (1994).

\bibitem{Bog}
J. Bognar, {\it Indefinite Inner Product Spaces}, Springer--Verlag, Berlin, 1974.

\bibitem{Gallo}
G. Galloway, {\it Closed timelike geodesics}, Trans. Amer. Math. Soc. {\bf 285} (1984), 379--384.

\bibitem{Tipler}
F. J. Tipler, {\it Existence of a closed timelike geodesic in Lorentz spaces}, Proc. Amer. Math. Soc. {\bf 76} (1979), 145--147.

\bibitem{Gayral}
V. Gayral, J. M. Gracia-Bondía, B. Iochum, T. Schücker and J.~C.~Varilly, {\it Moyal Planes are Spectral Triples}, Comm. Math. Phys. {\bf 246} (2004),  569--623, hep-th/0307241.

\bibitem{BS03}
A.N. Bernal and M. S\'anchez, {\it On smooth Cauchy hypersurfaces and Geroch's splitting theorem}, Comm. Math. Phys. {\bf 243} (2003), 461--470, gr-qc/0306108.

\bibitem{BS04}
A.N. Bernal and M. S\'anchez, {\it Smoothness of time functions and the metric splitting of globally hyperbolic spacetimes}, Comm. Math. Phys. {\bf 257} (2005), 43--50, gr-qc/0401112.

\bibitem{BS06}
A.N. Bernal and M. S\'anchez, {\it Further results on the smoothability of Cauchy hypersurfaces and Cauchy time functions}, Lett. Math. Phys. {\bf 77} (2006), 183--197, gr-qc/0512095.

\bibitem{F4}
N. Franco, {\it Lorentzian approach to noncommutative geometry}, Ph.D. thesis, Presses Universitaires de Namur, 2011, arXiv:1108.0592.

\bibitem{PZ}
G. N. Parfionov and R. R. Zapatrin, {\it Connes duality in Lorentzian geometry}, J. Math. Phys. {\bf 41} (2000), 7122--7128, gr-qc/9803090.

\bibitem{Mor}
V. Moretti, {\it Aspects of noncommutative Lorentzian geometry for globally hyperbolic spacetimes}, Rev. Math. Phys. {\bf 15} (2003), \mbox{1171--1217}, gr-qc/0203095.

\bibitem{Bes}
F. Besnard, {\it A noncommutative view on topology and order}, \mbox{J. Geom. Phys.} {\bf 59} 7 (2009), 861--875, math/0804.3551.

\bibitem{Antoine}
J.-P. Antoine and C. Trapani, {\it Partial Inner Product Spaces: Theory and Applications}, Lecture Notes in Mathematics {\bf 1986}, Springer--Verlag, Berlin, 2010.

\bibitem{VG88}
J. C. Várilly and J. M. Gracia-Bondía, {\it Algebras of distributions suitable for phase-space
quantum mechanics II: Topologies on the Moyal algebra}, J. Math. Phys. {\bf 29} (1988), 880--
887.

\bibitem{Beem}
J. K. Beem, P. E. Eherlich and K. L. Easley, {\it Global Lorentzian geometry},  2nd ed., in Monographs and Textbooks in Pure and Applied Mathematics {\bf 202}, Marcel Dekker, New York, 1996.

\bibitem{BraRo}
O. Bratteli and D. W. Robinson, {\it Operator algebras and quantum statistical mechanics. 1, C*- and W*-algebras symmetry groups decomposition of states}, second edition, Springer, New York, 1987.

\end{thebibliography}
\end{document}